\documentclass[english]{amsart}
\usepackage[T1]{fontenc}
\usepackage[latin9]{inputenc}
\usepackage{amsthm}
\usepackage{graphicx}
\usepackage{amssymb}

\makeatletter

\providecommand{\tabularnewline}{\\}

\numberwithin{equation}{section}
\numberwithin{figure}{section}
\theoremstyle{plain}
\newtheorem{thm}{Theorem}
  \theoremstyle{definition}
  \newtheorem{defn}[thm]{Definition}
  \theoremstyle{remark}
  \newtheorem{rem}[thm]{Remark}
  \theoremstyle{definition}
  \newtheorem{problem}[thm]{Problem}
  \theoremstyle{plain}
  \newtheorem{cor}[thm]{Corollary}
 \theoremstyle{definition}
  \newtheorem{example}[thm]{Example}
  \theoremstyle{plain}
  \newtheorem{algorithm}[thm]{Algorithm}
  \theoremstyle{plain}
  \newtheorem{lem}[thm]{Lemma}
  \theoremstyle{remark}
  \newtheorem{claim}[thm]{Claim}
  \theoremstyle{remark}
  \newtheorem{notation}[thm]{Notation}
  \theoremstyle{plain}
  \newtheorem{criterion}[thm]{Criterion}

\usepackage{amsmath, amssymb, graphics}

\newcommand{\mathsym}[1]{{}}

\makeatother

\usepackage{babel}

\begin{document}

\title{Asymptotic Solutions of Polynomial Equations with Exp-Log Coefficients}

\author{Adam Strzebo\'nski}

\curraddr{Wolfram Research Inc., 100 Trade Centre Drive, Champaign, IL 61820,
U.S.A.}

\email{adams@wolfram.com}
\begin{abstract}
We present an algorithm for computing asymptotic approximations of
roots of polynomials with exp-log function coefficients. The real
and imaginary parts of the approximations are given as explicit exp-log
expressions. We provide a method for deciding which approximations
correspond to real roots. We report on implementation of the algorithm
and present empirical data.
\end{abstract}
\maketitle

\section{Introduction}
\begin{defn}
The set of \emph{exp-log functions} is the smallest set of partial
functions $\mathbb{\mathbb{R}\rightarrow R}$ containing $\exp$,
$\log$, the identity function and the constant functions, closed
under addition, multiplication and composition of functions. 
\end{defn}
The domain $D(f)$ of an exp-log function $f$ is determined as follows:
\begin{enumerate}
\item the domain of $exp$, the identity function and the constant functions
is $\mathbb{R}$ and the domain of $log$ is $\mathbb{R}_{+}$,
\item $D(f+g)=D(fg)=D(f)\cap D(g)$,
\item $D(f(g))=g^{-1}(D(f))$.
\end{enumerate}
In particular, $D(f)$ is an open set and $f$ is $C^{\infty}$ in
$D(f)$.
\begin{rem}
The multiplicative inverse function\[
\mathbb{R}\setminus\{0\}\ni x\rightarrow1/x=x\,\exp(-\log(x^{2}))\in\mathbb{R}\]
and the real exponent power functions\[
\mathbb{R}_{+}\ni x\rightarrow x^{r}=\exp(r\log(x))\in\mathbb{R}\]
for $r\in\mathbb{R}$, are exp-log functions.
\end{rem}
The domain of an exp-log function consists of a finite number of open,
possibly unbounded, intervals and an exp-log function has a finite
number of real roots. An algorithm computing domains and isolating
intervals for real roots of exp-log functions is given in \cite{S21,S22}. 

We say that a partial function $f:\mathbb{R}\rightarrow\mathbb{C}$
is \emph{defined near infinity} if $D(f)\supseteq(c,\infty)$ for
some $c\in\mathbb{R}$. 
\begin{defn}
A \emph{Hardy field} \cite{B} is a set of germs at infinity of real-valued
functions that is closed under differentiation and forms a field under
addition and multiplication.\end{defn}
\begin{thm}
\cite{H1,H2} The germs at infinity of exp-log functions defined near
infinity form a Hardy field. 
\end{thm}
Let $P(x,y)=a_{n}(x)y^{n}+\ldots+a_{0}(x)$ where, for $0\leq i\leq n$,
$a_{i}(x)=u_{i}(x)+\imath v_{i}(x)$ and $u_{i}$ and $v_{i}$ are
exp-log functions defined near infinity ($\imath$ denotes the imaginary
unit). 
\begin{problem}
Describe the asymptotic behaviour of roots of $P$ in $y$ as $x$
tends to infinity.
\end{problem}
The following theorem \cite{R5,R6} shows that the problem is well
posed, that is the roots of $P$ in $y$ are $C^{\infty}$ functions
in $x$ defined near infinity.
\begin{thm}
If $H$ is a Hardy field, then there exists a Hardy field $K\supseteq H$
such that $K[\imath]$ is algebraically closed.\end{thm}
\begin{cor}
There exists a Hardy field $K$ such that $P$ has $n$ roots in $y$
(counted with multiplicities) in $K[\imath]$.\end{cor}
\begin{defn}
Let $f:\mathbb{R}\rightarrow\mathbb{C}$ be a partial function defined
near infinity. We say that partial functions $f_{1},\ldots,f_{m}:\mathbb{R}\rightarrow\mathbb{C}$
defined near infinity form an \emph{$m$-term asymptotic approximation}
of $f$ if, for $1\leq i<m$, $\lim_{x\rightarrow\infty}\frac{f_{i+1}(x)}{f_{i}(x)}=0$
and $\lim_{x\rightarrow\infty}\frac{f(x)-\sum_{i=1}^{m}f_{i}(x)}{f_{m}(x)}=0$. 

In this paper we present an algorithm which computes asymptotic approximations
of roots of $P$ in $y$. The approximations are given as exp-log
expressions. The algorithm makes use of the theory of {}``most rapidly
varying'' subexpressions developed in \cite{G} to compute limits
of exp-log functions. In fact our algorithm applies in the more general
case of \emph{MrvH} fields. The algorithm is based on a Newton polygon
technique \cite{N,W1,H3} extended to {}``series'' with arbitrary
real exponents. 

Algorithms given in \cite{Sh1,H3} solve the problem of finding asymptotic
solutions of polynomial equations in more general settings. We chose
to extend the algorithm of \cite{G} because we find it simpler to
implement and we can give a direct and elementary proof that the computed
expressions satisfy our (weaker) requirements.\end{defn}
\begin{example}
\label{exa:intro}Let $P(x,y)=y^{5}-\exp(x)y-\log(x)$. One-term asymptotic
approximations of roots of $P$ in $y$ computed with our algorithm
are $r_{1}=-\exp(-x)\log(x)$, $r_{2}=-\exp(-x)^{-1/4}$, $r_{3}=\exp(-x)^{-1/4}$,
$r_{4}=-\imath\exp(-x)^{-1/4}$, $r_{5}=\imath\exp(-x)^{-1/4}$. Let
us estimate the relative error $\varepsilon_{i}=\lvert\frac{r_{i}-r_{i}^{*}}{r_{i}}\rvert$
of the approximations, where $r_{i}^{*}$ is the exact root closest
to $r_{i}$. Using the bound \[
\lvert y-y^{*}\rvert\leq\lvert\frac{5P(x,y)}{\frac{\partial P}{\partial y}(x,y)}\rvert\]
on the distance from $y$ to the closest root $y^{*}$ of $P(x,y)$,
after simplifications valid for $x>1$, we get $\varepsilon_{1}\leq\frac{5\log(x)^{4}}{\exp(5x)-5\log(x)^{4}}$,
and for $i\neq1$, $\varepsilon_{i}\leq\frac{5}{4}\log(x)\exp(-\frac{5}{4}x)$.
Both bounds tend to zero as $x$ tends to infinity and are decreasing
for $x>\exp(W_{0}(\frac{4}{5}))=1.63\ldots$, where $W_{0}$ is the
principal branch of the Lambert W function. Evaluating the bounds
at $x=10$ we get $\log_{10}\varepsilon_{1}\leq-19.5$ and $\log_{10}\varepsilon_{i}\leq-4.96$
for $i\neq1$. For $x=1000$ we get $\log_{10}\varepsilon_{1}\leq-2167$
and $\log_{10}\varepsilon_{i}\leq-541$ for $i\neq1$. This shows
that we can obtain approximations of roots of $P(1000,y)$ to $500$
digits of precision by evaluating the asymptotic approximations. The
evaluation takes $2$ ms. For comparison, direct computation of roots
of $P(1000,y)$ to $500$ digits of precision takes $68$ ms. Figure
\ref{fig:intro} shows the asymptotic approximations of the real roots
of $P(x,y)$ (dashed curves) and the exact roots (solid green curves).

\begin{figure}

\caption{\label{fig:intro}Real roots from Example \ref{exa:intro}}
\includegraphics[width=\columnwidth, trim = 0mm 0mm 0mm 0mm, clip]{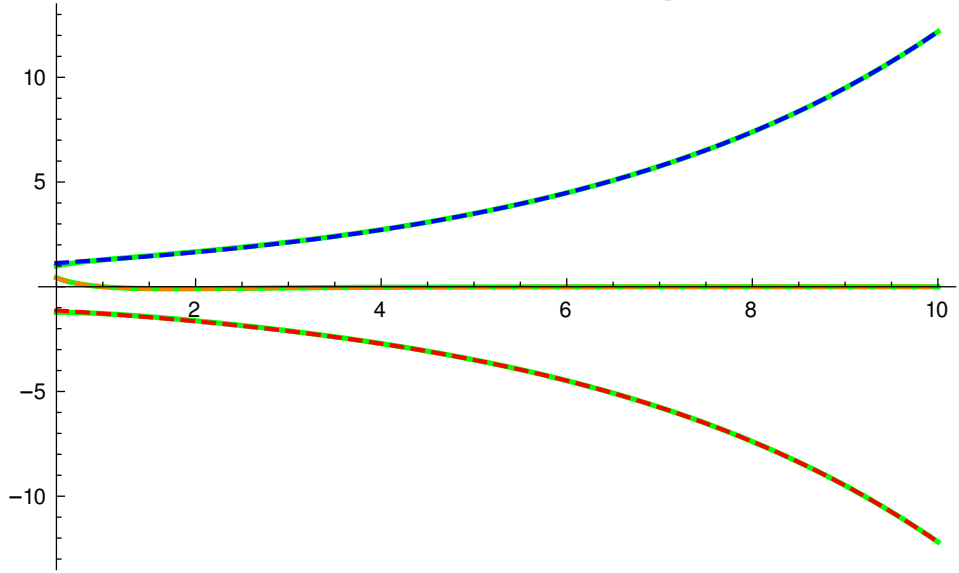}

\end{figure}

\end{example}

\section{Most rapidly varying subexpressions}

In this section we give a very brief summary of terminology and facts
necessary for formulating Algorithm \ref{alg:MrvApprox}. We will
use the algorithm to compute approximations of coefficients of $P$
in terms of a {}``most rapidly varying'' subexpression present in
the coefficients. The algorithm is based on the algorithm MrvLimit
described in \cite{G}. For a more detailed introduction and proofs
of the stated facts see \cite{G}. 
\begin{defn}
The set of \emph{exp-log expressions} with coefficients in a computable
field $C\subseteq\mathbb{R}$ is defined recursively as follows:
\begin{enumerate}
\item elements of $C$ and the variable $x$ are exp-log expressions,
\item if $f$ and $g$ are exp-log expressions, so are $f+g$, $f\cdot g$
and $\frac{f}{g}$,
\item if $f$ is an exp-log expression and $c\in C$, then $\exp(f)$, $\log(f)$,
and $f^{c}$ are exp-log expressions.
\end{enumerate}
\end{defn}
Each exp-log expression represents an exp-log function, however the
same function may be represented by many different expressions. In
the following, when we refer to the domain, point values, and limits
of an exp-log expression, we mean the domain, point values, and limits
of the corresponding exp-log function. 
\begin{defn}
Let $E_{\infty}(x)$ be the set of exp-log expressions $f$ such that,
for some $c\in\mathbb{R}$, $(c,\infty)\subseteq D(f)$ and either
$f=0$ as an expression or $f$ is nonzero on $(c,\infty)$.\end{defn}
\begin{rem}
Note that we exclude from $E_{\infty}(x)$ expressions that are identically
zero in a neighbourhood of infinity, but are not explicitly zero e.g.
$f=\exp(\log((x-c)^{2})/2)-x+c$. The algorithm \emph{ExpLogRootIsolation}
of \cite{S21} can be used to check whether a given exp-log expression
is defined near infinity and to detect expressions that are identically
zero in a neighbourhood of infinity and replace them with explicit
zeros. \emph{ExpLogRootIsolation} requires a zero test algorithm for
elementary constants. Termination of the currently known zero test
algorithm relies on Schanuel's conjecture \cite{R3,S22}.
\end{rem}
The germs at infinity of functions represented by elements of $E_{\infty}(x)$
form the Hardy field of exp-log functions defined near infinity.
\begin{thm}
If $f$ and $g$ are nonzero elements of a Hardy field, then the limit
\[
\lim_{x\rightarrow\infty}\frac{\log\lvert f(x)\rvert}{\log\lvert g(x)\rvert}\]
exists (in $\overline{\mathbb{R}}=\mathbb{R}\cup\{-\infty,\infty\}$).
Moreover, if $\lim_{x\rightarrow\infty}g(x)=0$ and\[
\lim_{x\rightarrow\infty}\frac{\log\lvert f(x)\rvert}{\log\lvert g(x)\rvert}=0\]
then, for any $e>0$, \[
\lim_{x\rightarrow\infty}f(x)g(x)^{e}=0\]

\end{thm}
The theorem follows from the results in section 3.1.2 of \cite{G}.

Following \cite{G}, we say that $g$ is more rapidly varying than
$f$, or $g$ is in a higher comparability class than $f$, if \[
\lim_{x\rightarrow\infty}\frac{\log\lvert f(x)\rvert}{\log\lvert g(x)\rvert}=0\]
and we denote it $f\prec g$. We say that $f$ and $g$ have the same
order of variation, or $f$ and $g$ are in the same comparability
class, if \[
\lim_{x\rightarrow\infty}\frac{\log\lvert f(x)\rvert}{\log\lvert g(x)\rvert}\in\mathbb{R}\setminus\{0\}\]
and we denote it $f\asymp g$. We will also use $f\preceq g$ to denote
$f\prec g\vee f\asymp g$. $\omega$ is a \emph{most rapidly varying
subexpression} of $f$ if $\omega$ is a subexpression of $f$ and
no subexpression of $f$ is more rapidly varying than $\omega$. Let
$mrv(f)$ be the set of most rapidly varying subexpressions of $f$.
We will write $mrv(f)\prec g$ (resp. $mrv(f)\asymp g$) if for all
$\omega\in mrv(f)$, $\omega\prec g$ (resp. $\omega\asymp g$). Let
$mrv(f_{1},\ldots,f_{n})$ be the set of $\omega$ such that $\omega$
is a subexpression of some $f_{i}$ and no subexpression of any $f_{j}$
is more rapidly varying than $\omega$, that is $mrv(f_{1},\ldots,f_{n})$
the $\max$ of $mrv(f_{1}),\ldots,mrv(f_{n})$ (as in Algorithm 3.12
of \cite{G}).

To prove termination of our algorithm we use the notion of size of
an exp-log expression defined in \cite{G}, section 3.4.1. For an
exp-log expression $f$, let $S(f)$ be the set of subexpressions
of $f$ defined by the following conditions.
\begin{enumerate}
\item If $f$ does not contain the variable $x$, then $S(f)=\emptyset$.
\item If $f=x$, then $S(f)=\{x\}$.
\item If $f=g+h$, $f=gh$, or $f=\frac{g}{h}$, then $S(f)=S(g)\cup S(h)$.
\item If $f=g^{c}$ then $S(f)=S(g)$.
\item If $f=\exp g$ or $f=\log g$ then $S(f)=\{f\}\cup S(g)$.
\end{enumerate}
Then $Size(f)$ is defined as the cardinality of $S(f)$.

Let $\exp^{k}$ (resp. $\log^{k}$) denote $k$ times iterated exponential
(resp. logarithm), and for $f\in E_{\infty}(x)$  let $f^{\uparrow k}$
(resp. $f^{\downarrow k}$) denote $f$ with $x$ replaced with $\exp^{k}(x)$
(resp. $\log^{k}(x)$). The following algorithm computes approximations
of elements of a finite subset of $E_{\infty}(x)$ in terms of their
most rapidly varying subexpression.
\begin{algorithm}
\label{alg:MrvApprox}(MrvApprox)\\
Input: $a_{0},\ldots,a_{n}\in E_{\infty}(x)$ such that $\sum_{i=0}^{n}Size(a_{i})>0$.\\
Output: $\omega\in E_{\infty}(x)$, $b_{0},\ldots,b_{n}\in E_{\infty}(x)$,
$e_{0},\ldots,e_{n}\in\mathbb{R}\cup\{\infty\}$, $d>0$, and $k\in\mathbb{Z}_{\geq0}$
such that
\begin{itemize}
\item $\omega>0$ and $\lim_{x\rightarrow\infty}\omega=0$,
\item $mrv(b_{0},\ldots,b_{n})\prec\omega$,
\item if $a_{i}=0$ then $b_{i}=0$ and $e_{i}=\infty$,
\item if $a_{i}\neq0$ then $e_{i}\in\mathbb{R}$ and $\lim_{x\rightarrow\infty}\omega^{-(e_{i}+d)}(a_{i}^{\uparrow k}-b_{i}\omega^{e_{i}})=0$,
\item $\sum_{i=0}^{n}Size(b_{i})<\sum_{i=0}^{n}Size(a_{i})$.
\end{itemize}
\end{algorithm}
The algorithm proceeds in a very similar manner to the algorithm MrvLimit
described in \cite{G}. First, it finds the set $\Omega=mrv(a_{0},\ldots,a_{n})$.
If $x\in\Omega$ the algorithm replaces $x$ with $\exp(x)$ in $a_{0},\ldots,a_{n}$
and recomputes $\Omega$ until $x\notin\Omega$. $k$ is the number
of replacements performed in this step. Then the algorithm picks $\omega$
such that $\omega$ or $1/\omega$ belongs to $\Omega$, $\omega>0$,
and $\lim_{x\rightarrow\infty}\omega=0$, and rewrites all elements
of $\Omega$ in terms of $\omega$. If $a_{i}^{\uparrow k}$ contains
a subexpression in the same comparability class as $\omega$, let
$b_{i}\omega^{e_{i}}$ be the first term of $Series(a_{i}^{\uparrow k},\omega)$
(as in section 3.3.3 of \cite{G}), and let $d_{i}>0$ be the difference
between the exponents of $\omega$ in the second and in the first
term of the series ($d_{i}=\infty$ if $a_{i}^{\uparrow k}=b_{i}\omega^{e_{i}}$).
If $a_{i}^{\uparrow k}$ does not contain subexpressions in the same
comparability class as $\omega$, then $b_{i}=a_{i}^{\uparrow k}$,
$e_{i}=0$, and $d_{i}=\infty$. In both cases $mrv(b_{i})\prec\omega$.
Pick $0<d<\min_{0\leq i\leq n}d_{i}$. Then $\omega^{-(e_{i}+d)}(a_{i}^{\uparrow k}-b_{i}\omega^{e_{i}})$
is either $0$ or a power series in $\omega$ with positive exponents
and coefficients in a lower comparability class than $\omega$, hence
$\lim_{x\rightarrow\infty}\omega^{-(e_{i}+d)}(a_{i}^{\uparrow k}-b_{i}\omega^{e_{i}})=0$.
Section 3.4.1 of \cite{G} proves that $Size(b_{i})\leq Size(a_{i})$,
with the strict inequality if $a_{i}^{\uparrow k}$ contains a subexpression
in the same comparability class as $\omega$. This shows that the
last requirement is satisfied.

\section{Root continuity}

To prove correctness of our algorithm we need a polynomial root continuity
lemma that does not assume fixed degree of the polynomial. The lemma
is very similar to Theorem 1 of \cite{Z}, except that our version
provides explicit bounds.
\begin{lem}
\label{lem:Roots}Let\[
p=a_{n}z^{n}+\ldots+a_{0}=a_{n}(z-r_{1})\cdot\ldots\cdot(z-r_{n})\in\mathbb{C}[z]\]
where $n\geq1$ and $a_{n}\neq0$. Let $\Gamma=\max_{1\leq i\leq n}\lvert r_{i}\rvert$
and $\Delta=\min_{r_{i}\neq r_{j}}\lvert r_{i}-r_{j}\rvert$ ($\Delta=\infty$
if $r_{1}=\ldots=r_{n}$). 

Suppose that $m\geq n$, $0<\epsilon<\min(1,\frac{1}{\Gamma+1},\frac{\Delta}{2})$,
and $0<\delta<\lvert a_{n}\rvert\frac{(1-\epsilon)\epsilon^{m+n}}{1-\epsilon^{m+1}}$.

Then for every \[
q=b_{m}z^{m}+\ldots+b_{0}\in\mathbb{C}[z]\]
such that $b_{m}\neq0$, for $0\leq i\leq n$, $\lvert b_{i}-a_{i}\rvert<\delta$,
and, for $n+1\leq i\leq m$, $\lvert b_{i}\rvert<\delta$, we have\[
q=b_{m}(z-s_{1})\cdot\ldots\cdot(z-s_{m})\]
 for $1\leq i\leq n$, $\lvert s_{i}-r_{i}\rvert<\epsilon$, and for
$n+1\leq i\leq m$, $\lvert s_{i}\rvert>1/\epsilon$.\end{lem}
\begin{proof}
Let $C_{0}=\{c\,:\,\lvert c\rvert=1/\epsilon\}$, $D_{0}=\{c\,:\,\lvert c\rvert\leq1/\epsilon\}$,
and, for $1\leq i\leq n$, let $C_{i}=\{c\,:\,\lvert c-r_{i}\rvert=\epsilon\}$,
and $D_{i}=\{c\,:\,\lvert c-r_{i}\rvert\leq\epsilon\}$. Then, for
$1\leq i,j\leq n$, $D_{i}$ and $D_{j}$ are either identical or
disjoint, $D_{i}$ is contained in the interior of $D_{0}$, $D_{i}$
contains exactly one of the distinct roots of $p$, and $D_{0}$ contains
all roots of $p$. If $c\in C_{i}$ for some $0\leq i\leq n$, then
$\lvert p(c)\rvert\geq\lvert a_{n}\rvert\epsilon^{n}$ and\[
\lvert q(c)-p(c)\rvert\leq\sum_{k=0}^{n}\lvert b_{k}-a_{k}\rvert\lvert z\rvert^{k}+\sum_{k=n+1}^{m}\lvert b_{k}\rvert\lvert z\rvert^{k}\leq\delta\sum_{k=0}^{m}\frac{1}{\epsilon^{k}}\]
We have\[
\delta\sum_{k=0}^{m}\frac{1}{\epsilon^{k}}<\lvert a_{n}\rvert\frac{(1-\epsilon)\epsilon^{m+n}}{1-\epsilon^{m+1}}\frac{1-1/\epsilon^{m+1}}{1-1/\epsilon}=\lvert a_{n}\rvert\epsilon^{n}\]
Hence $\lvert q(c)-p(c)\rvert<\lvert p(c)\rvert$. By Rouche's theorem,
for $0\leq i\leq n$, the number of roots of $q$ in $D_{i}$ equals
the number of roots of $p$ in $D_{i}$, which concludes the proof.
\end{proof}

\section{The main algorithm}

Let $E_{\infty}^{\mathbb{C}}(x)=\{u+\imath v\,:\, u,v\in E_{\infty}(x)\}$.
This section presents the main algorithm computing asymptotic approximations
of roots of polynomials $P\in E_{\infty}^{\mathbb{C}}(x)[y]$. 

Let us first describe a straightforward generalization of Algorithm
\ref{alg:MrvApprox} to inputs in $E_{\infty}^{\mathbb{C}}(x)$. We
extend $mrv$ and $Size$ to $E_{\infty}^{\mathbb{C}}(x)$ by defining
$mrv(u_{1}+\imath v_{1},\ldots,u_{n}+\imath v_{n})=mrv(u_{1},v_{1},\ldots,u_{n},v_{n})$
and $Size(u+\imath v)=Size(u)+Size(v)$. We say that $u+\imath v\prec\omega$
(resp. $u+\imath v\asymp\omega$) if $\lvert u+\imath v\rvert=(u^{2}+v^{2})^{1/2}\prec\omega$
(resp. $\lvert u+\imath v\rvert\asymp\omega$). If $a=u+\imath v$
then $a^{\uparrow k}:=u^{\uparrow k}+\imath v^{\uparrow k}$ and $a^{\downarrow k}:=u^{\downarrow k}+\imath v^{\downarrow k}$.
\begin{algorithm}
\label{alg:MrvApproxC}(MrvApproxC)\\
Input: $a_{0},\ldots,a_{n}\in E_{\infty}^{\mathbb{C}}(x)$ such
that $\sum_{i=0}^{n}Size(a_{i})>0$.\\
Output: $\omega\in E_{\infty}(x)$, $b_{0},\ldots,b_{n}\in E_{\infty}^{\mathbb{C}}(x)$,
$e_{0},\ldots,e_{n}\in\mathbb{R}\cup\{\infty\}$, $d>0$, and $k\in\mathbb{Z}_{\geq0}$
such that
\begin{itemize}
\item $\omega>0$ and $\lim_{x\rightarrow\infty}\omega=0$,
\item $mrv(b_{0},\ldots,b_{n})\prec\omega$,
\item if $a_{i}=0$ then $b_{i}=0$ and $e_{i}=\infty$,
\item if $a_{i}\neq0$ then $e_{i}\in\mathbb{R}$ and $\lim_{x\rightarrow\infty}\omega^{-(e_{i}+d)}(a_{i}^{\uparrow k}-b_{i}\omega^{e_{i}})=0$,
\item $\sum_{i=0}^{n}Size(b_{i})<\sum_{i=0}^{n}Size(a_{i})$.\end{itemize}
\begin{enumerate}
\item Let $a_{i}=u_{i}+\imath v_{i}$. Call Algorithm \ref{alg:MrvApprox}
with $u_{0},v_{0},\ldots,u_{n},v_{n}$ as input, obtaining $\omega\in E_{\infty}(x)$,
$\hat{u}_{0},\hat{v}_{0},\ldots,\hat{u}_{n},\hat{v}_{n}\in E_{\infty}(x)$,
$e_{u,0},e_{v,0},\ldots,e_{u,n}e_{v,n}\in\mathbb{R}$, $\hat{d}>0$,
and $k\in\mathbb{Z}_{\geq0}$. 
\item For $0\leq i\leq n$, put $e_{i}:=\min(e_{u,i},e_{v,i})$ and\[
b_{i}:=\begin{cases}
\hat{u}_{i} & e_{u,i}<e_{v,i}\\
\hat{u}_{i}+\imath\hat{v}_{i} & e_{u,i}=e_{v,i}\\
\imath\hat{v}_{i} & e_{u,i}>e_{v,i}\end{cases}\]

\item Pick $d$ such that $0<d<\hat{d}$ and $d<\lvert e_{u,i}-e_{v,i}\rvert$
for all $i$ such that $e_{u,i}\neq e_{v,i}$. 
\item Return $\omega$, $b_{0},\ldots,b_{n}$, $e_{0},\ldots,e_{n}$, $d$,
and $k$.
\end{enumerate}
\end{algorithm}
\begin{proof}
To prove that the output of Algorithm \ref{alg:MrvApproxC} satisfies
the required conditions we need to prove that if $a_{i}\neq0$ then\[
\lim_{x\rightarrow\infty}\omega^{-(e_{i}+d)}(a_{i}^{\uparrow k}-b_{i}\omega^{e_{i}})=0\]
 The other conditions follow directly from the definitions and the
properties of the output of Algorithm \ref{alg:MrvApprox}. 

Suppose that $e_{u,i}<e_{v,i}$. Then $e_{i}=e_{u,i}$ and \[
\omega^{-(e_{i}+d)}(a_{i}^{\uparrow k}-b_{i}\omega^{e_{i}})=\omega^{-(e_{u,i}+d)}(u_{i}^{\uparrow k}-\hat{u}_{i}\omega^{e_{u,i}})+\imath\omega^{-(e_{u,i}+d)}v_{i}^{\uparrow k}\]
We have\begin{align*}
 & \lim_{x\rightarrow\infty}\omega^{-(e_{u,i}+d)}(u_{i}^{\uparrow k}-\hat{u}_{i}\omega^{e_{u,i}})=\\
 & \lim_{x\rightarrow\infty}\omega^{\hat{d}-d}\omega^{-(e_{u,i}+\hat{d})}(u_{i}^{\uparrow k}-\hat{u}_{i}\omega^{e_{u,i}})=0\end{align*}
and\begin{align*}
 & \lim_{x\rightarrow\infty}\omega^{-(e_{u,i}+d)}v_{i}^{\uparrow k}=\\
 & \lim_{x\rightarrow\infty}\omega^{(e_{v,i}-e_{u,i})+(\hat{d}-d)}\omega^{-(e_{v,i}+\hat{d})}(v_{i}^{\uparrow k}-\hat{v}_{i}\omega^{e_{v,i}})+\\
 & \hat{v}_{i}\omega^{(e_{v,i}-e_{u,i})-d}=0\end{align*}
since $(e_{v,i}-e_{u,i})+(\hat{d}-d)>0$, $(e_{v,i}-e_{u,i})-d>0$,
and $\hat{v}_{i}\prec\omega$. Cases $e_{u,i}=e_{v,i}$ and $e_{u,i}>e_{v,i}$
can be proven in a similar manner.
\end{proof}
Let $P(x,y)=a_{n}(x)y^{n}+\ldots+a_{0}(x)\in E_{\infty}^{\mathbb{C}}(x)[y]$.
W.l.o.g. we may assume that $a_{n}$ and $a_{0}$ are not identically
zero.

Suppose that $\sum_{i=0}^{n}Size(a_{i})>0$ i.e. $P(x,y)$ depends
on $x$. Let $\omega\in E_{\infty}(x)$, $b_{0},\ldots,b_{n}\in E_{\infty}^{\mathbb{C}}(x)$,
$e_{0},\ldots,e_{n}\in\mathbb{R}\cup\{\infty\}$, $d>0$, and $k\in\mathbb{Z}_{\geq0}$
be the output of Algorithm \ref{alg:MrvApproxC} for $a_{0},\ldots,a_{n}$,
and let $Q(x,y)=a_{n}^{\uparrow k}(x)y^{n}+\ldots+a_{0}^{\uparrow k}(x)$.

Let $K$ be a Hardy field containing germs at infinity of exp-log
functions defined near infinity, such that $K[\imath]$ is algebraically
closed. Let $\alpha\in K[\imath]$ be a root of $Q$. Since $\lvert\alpha\rvert\in K$,
the limit $\lim_{x\rightarrow\infty}\frac{\log\lvert\alpha\rvert}{\log\lvert\omega\rvert}=\gamma$
exists. 
\begin{claim}
$\gamma\in\mathbb{R}$. \end{claim}
\begin{proof}
Suppose that $\lvert\gamma\rvert=\infty$. Then\begin{equation}
0=\lim_{x\rightarrow\infty}\frac{Q(x,\alpha)}{a_{n}^{\uparrow k}\alpha^{n}}=\lim_{x\rightarrow\infty}1+\sum_{i=1}^{n}\frac{a_{n-i}^{\uparrow k}}{a_{n}^{\uparrow k}\alpha^{i}}\label{eq:gamma-proof}\end{equation}
Since $\omega\prec\lvert\alpha\rvert$, $\frac{a_{n-i}^{\uparrow k}}{a_{n}^{\uparrow k}}\prec\lvert\alpha\rvert$,
and so either $\lim_{x\rightarrow\infty}\lvert\alpha\rvert=\infty$
and\[
\lim_{x\rightarrow\infty}\lvert\sum_{i=1}^{n}\frac{a_{n-i}^{\uparrow k}}{a_{n}^{\uparrow k}\alpha^{i}}\rvert=0\]
or $\lim_{x\rightarrow\infty}\lvert\alpha\rvert=0$ and\[
\lim_{x\rightarrow\infty}\lvert\sum_{i=1}^{n}\frac{a_{n-i}^{\uparrow k}}{a_{n}^{\uparrow k}\alpha^{i}}\rvert=\infty\]
Both cases contradict equation (\ref{eq:gamma-proof}). 
\end{proof}
Let $I=\{i\,:\,0\leq i\leq n\wedge a_{i}\neq0\}$. If $i\in I$, put\[
c_{i}=\omega^{-(e_{i}+d)}(a_{i}^{\uparrow k}-b_{i}\omega^{e_{i}})\]
 Then $a_{i}^{\uparrow k}=b_{i}\omega^{e_{i}}+c_{i}\omega^{e_{i}+d}$
and $\lim_{x\rightarrow\infty}c_{i}=0$. Put $\beta:=\frac{\alpha}{\omega^{\gamma}}$.
Then\[
\gamma=\lim_{x\rightarrow\infty}\frac{\log\lvert\beta\omega^{\gamma}\rvert}{\log\lvert\omega\rvert}=\gamma+\lim_{x\rightarrow\infty}\frac{\log\lvert\beta\rvert}{\log\lvert\omega\rvert}\]
and hence $\beta\prec\omega$. We have\[
Q(x,\alpha)=\sum_{i\in I}(b_{i}\beta^{i}\omega^{e_{i}+\gamma i}+c_{i}\beta^{i}\omega^{e_{i}+d+\gamma i})\]
Let $\mu=\min_{i\in I}e_{i}+\gamma i$, let $J=\{i\in I\,:\, e_{i}+\gamma i=\mu\}$,
and let \[
\mu<\nu<\min(\mu+d,\min_{i\in I\setminus J}e_{i}+\gamma i)\]
Then\[
0=\omega^{-\nu}Q(x,\alpha)=\omega^{\mu-\nu}\sum_{i\in J}b_{i}\beta^{i}+\sum_{i\in I\setminus J}b_{i}\beta^{i}\omega^{e_{i}+\gamma i-\nu}+\sum_{i\in I}c_{i}\beta^{i}\omega^{e_{i}+d+\gamma i-\nu}\]
 As $x$ tends to infinity, all terms in the last two sums tend to
zero, hence \[
\lim_{x\rightarrow\infty}\omega^{\mu-\nu}\sum_{i\in J}b_{i}\beta^{i}=0\]
 Since $\beta\prec\omega$ and $\mu-\nu<0$, the cardinality of $J$
must be at least $2$. Consider the subset $A=\{(i,e_{i})\,:\, i\in I\}$
of $\mathbb{R}^{2}$ with coordinates denoted $(x_{1},x_{2})$. Then
the line $x_{2}=-\gamma x_{1}+\mu$ passes through the points $\{(i,e_{i})\,:\, i\in J\}$
and all the other points of $A$ lie above this line. This means that
$-\gamma$ is the slope of one of the segments that form the lower
part of the boundary of the convex hull of $A$.

Let $R_{\gamma}(x,z)=\sum_{i\in J}b_{i}z^{i}$ and $Q_{\gamma}(x,z)=\omega^{-\mu}Q(x,z\omega^{\gamma})$.
We have\[
Q_{\gamma}(x,z)=R_{\gamma}(x,z)+\omega^{\nu-\mu}(\sum_{i\in I\setminus J}b_{i}\omega^{e_{i}+\gamma i-\nu}z^{i}+\sum_{i\in I}c_{i}\omega^{e_{i}+d+\gamma i-\nu}z^{i})\]

Let $\rho_{1},\ldots,\rho_{t}\in K[\imath]\setminus\{0\}$ be the
nonzero roots of $R_{\gamma}(x,z)$ listed with multiplicities.
\begin{claim}
There exist roots $\beta_{1},\ldots,\beta_{t}\in K[\imath]$ of $Q_{\gamma}(x,z)$
such that, for sufficiently large $x$, for $1\leq j\leq t$, $\lvert\beta_{j}-\rho_{j}\rvert<\omega^{\eta}$,
where $\eta=\frac{\nu-\mu}{4(n+1)}>0$. \end{claim}
\begin{proof}
Let $\Gamma(x)$ be the maximum of absolute values of roots of $R_{\gamma}(x,z)$
and let $\Delta(x)$ be the minimum distance between two distinct
roots of $R_{\gamma}(x,z)$ ($\Delta(x)=\infty$ if all roots of $R_{\gamma}(x,z)$
are equal). Put $\delta(x)=\omega(x)^{(\nu-\mu)/2}$ and $\epsilon(x)=\delta(x)^{1/(2n+2)}$.
For sufficiently large $x$, $R_{\gamma}(x,z)$ has a fixed number
of distinct roots in $z$, equal to its number of distinct roots in
$K[\imath]$. $\Gamma(x)$ can be bounded from above by a rational
function in absolute values of $b_{i}$, for $i\in J$, and, since
the coefficients in $z$ of $R_{1}(x,z)/g.c.d.(R_{1}(x,z),\frac{\partial}{\partial z}R_{1}(x,z))$
are rational functions of $b_{i}$, $\Delta(x)$ can be bounded from
below by an expression constructed from $b_{i}$ using rational operations,
square roots and absolute value (see e.g. \cite{M}, Theorem 5). Since
$mrv(b_{i})\prec\omega$, for sufficiently large $x$, $0<\epsilon<\min(1,\frac{1}{\Gamma+1},\frac{\Delta}{2})$.
Let $s$ be the degree of $R_{\gamma}$ in $z$. For sufficiently
large $x$, \[
\lvert b_{s}\rvert\frac{(1-\epsilon)\epsilon^{n+s}}{1-\epsilon^{n+1}}>\lvert b_{s}\rvert\epsilon^{n+s+1}>\epsilon^{2n+2}=\delta\]
The coefficients at $z^{i}$, for $0\leq i\leq n$, of $Q_{\gamma}(x,z)-R_{\gamma}(x,z)$
have the form $\omega^{\nu-\mu}\xi_{i}$ and $\lim_{x\rightarrow\infty}\xi_{i}=0$,
hence, for sufficiently large $x$, the absolute value of each of
these coefficients is less than $\delta$. By Lemma \ref{lem:Roots},
there exist roots $\beta_{1},\ldots,\beta_{t}\in K[\imath]$ of $Q_{\gamma}(x,z)$
such that, for sufficiently large $x$, for $1\leq j\leq t$, $\lvert\beta_{j}-\rho_{j}\rvert<\epsilon=\omega^{\eta}$.
\end{proof}
Fix $1\leq j\leq t$ and let $\alpha_{j}=\beta_{j}\omega^{\gamma}$.
Then $Q(x,\alpha_{j})=\omega^{\mu}Q_{\gamma}(x,\beta_{j})=0$. Suppose
that $f_{1},\ldots,f_{m}$ form an $m$-term asymptotic approximation\emph{
}of\emph{ }$\rho_{j}$ such that, for $1\leq i\leq m$, $mrv(f_{i})\prec\omega$.
For $1\leq i\leq m$, put $g_{i}=f_{i}\omega^{\gamma}$. We have\[
\frac{\alpha_{j}-\sum_{i=1}^{m}g_{i}}{g_{m}}=\frac{\beta_{j}-\sum_{i=1}^{m}f_{i}}{f_{m}}=\frac{\beta_{j}-\rho_{j}}{f_{m}}+\frac{\rho_{j}-\sum_{i=1}^{m}f_{i}}{f_{m}}\]
For sufficiently large $x$, $\lvert\beta_{j}-\rho_{j}\rvert<\omega^{\eta}$.
Since $f_{m}\prec\omega$, $\lim_{x\rightarrow\infty}\frac{\omega^{\eta}}{f_{m}}=0$.
Hence, $\lim_{x\rightarrow\infty}\frac{\alpha_{j}-\sum_{i=1}^{m}g_{i}}{g_{m}}=0$,
and so $g_{1},\ldots,g_{m}$ form an $m$-term asymptotic approximation\emph{
}of\emph{ }$\alpha_{j}$. Since for any $f\in K[\imath]$\[
\lim_{x\rightarrow\infty}f(x)=\lim_{x\rightarrow\infty}f(\log^{k}(x))\]
$g_{1}^{\downarrow k},\ldots,g_{m}^{\downarrow k}$ form an $m$-term
asymptotic approximation of the root $\alpha_{j}(\log^{k}(x))$ of
$P$.

Let us now consider the case where we have found an exact solution
$\rho_{j}\in E_{\infty}^{\mathbb{C}}(x)$ of $R_{\gamma}(x,z)$. To
simplify the description of the case let us make the following rather
technical definition.
\begin{defn}
Let $\omega\in E_{\infty}(x)$, $b_{0},\ldots,b_{n}\in E_{\infty}^{\mathbb{C}}(x)$,
$e_{0},\ldots,e_{n}\in\mathbb{R}\cup\{\infty\}$, $d>0$, and $k\in\mathbb{Z}_{\geq0}$
be the result of applying Algorithm \ref{alg:MrvApproxC} to $a_{0},\ldots,a_{n}$.
We will call a root $\alpha(\log^{k}(x))$ of $P$ \emph{asymptotically
small} if $\lim_{x\rightarrow\infty}\alpha=0$ and $\alpha\asymp\omega$
(in other words, $\alpha=\beta\omega^{\gamma}$ with $\gamma>0$).
\end{defn}
Suppose that we have found an exact solution $\rho_{j}\in E_{\infty}^{\mathbb{C}}(x)$
of $R_{\gamma}(x,z)$ of multiplicity $\sigma_{j}$. Then there exist
exactly $\sigma_{j}$ roots $\beta_{1,j},\ldots,\beta_{\sigma_{j},j}\in K[\imath]$
of $Q_{\gamma}(x,z)$ such that, for sufficiently large $x$, for
$1\leq\iota\leq\sigma_{j}$, $\lvert\beta_{\iota,j}-\rho_{j}\rvert<\omega^{\eta}$.
Hence, there exist exactly $\sigma_{j}$ roots $\alpha_{\iota,j}=\beta_{\iota,j}\omega^{\gamma}$
of $Q(x,y)$ such that, for sufficiently large $x$, for $1\leq\iota\leq\sigma_{j}$,
$\lvert\alpha_{\iota,j}-\rho_{j}\omega^{\gamma}\rvert<\omega^{\gamma+\eta}$.
The mapping $\varphi:\zeta\rightarrow\rho_{j}\omega^{\gamma}+\omega^{\gamma}\zeta$
is a bijection between the roots $\zeta$ of $Q(x,\rho_{j}\omega^{\gamma}+\omega^{\gamma}y)$
such that, for sufficiently large $x$, $\lvert\zeta\rvert<\omega^{\eta}$
and the roots $\varphi(\zeta)$ of $Q(x,y)$ such that $\lvert\varphi(\zeta)-\rho_{j}\omega^{\gamma}\rvert<\omega^{\gamma+\eta}$.
Since $\nu>\mu$ can be chosen arbitrarily close to $\mu$, $\eta$
can be arbitrarily small. Therefore the mapping $\xi\rightarrow(\rho_{j}\omega^{\gamma})^{\downarrow k}+(\omega^{\gamma})^{\downarrow k}\xi$
is a bijection between the roots $\xi$ of $P(x,(\rho_{j}\omega^{\gamma})^{\downarrow k}+(\omega^{\gamma})^{\downarrow k}y)$
that are identically zero or asymptotically small and the roots $\alpha_{1,j}(\log^{k}(x)),\ldots,\alpha_{\sigma_{j},j}(\log^{k}(x))$
of $P(x,y)$.

The above discussion suggests the following procedure for finding
asymptotic approximations of roots of $P$. Use Algorithm \ref{alg:MrvApproxC}
for $a_{0},\ldots,a_{n}$, to find $\omega\in E_{\infty}(x)$, $b_{0},\ldots,b_{n}\in E_{\infty}^{\mathbb{C}}(x)$,
$e_{0},\ldots,e_{n}\in\mathbb{R}\cup\{\infty\}$, $d>0$, and $k\in\mathbb{Z}_{\geq0}$.
Compute the values of $\gamma$ such that $-\gamma$ is the slope
of one of the segments that form the lower part of the boundary of
the convex hull of $A=\{(i,e_{i})\,:\, i\in I\}$. For each $\gamma$
find $R_{\gamma}(x,z)$ and call the procedure recursively to find
asymptotic approximations of roots of $R_{\gamma}$. Finally, obtain
asymptotic approximations of roots of $P$ by multiplying the terms
of asymptotic approximations of roots of $R_{\gamma}$ by $\omega^{\gamma}$
and replacing $x$ with $\log^{k}(x)$. The following algorithm formalizes
this procedure, handles the base case, and the case where we get an
exact solution with less than the requested $m$ terms. 
\begin{notation}
We use the notation $\sqcup$ for joining lists, that is\[
(f_{1},\ldots,f_{l})\sqcup(g_{1},\ldots,g_{m})=(f_{1},\ldots,f_{l},g_{1},\ldots,g_{m})\]
For a list $F=(f_{1},\ldots,f_{l})$ of expressions in $E_{\infty}^{\mathbb{C}}(x)$
let\[
F^{\downarrow k}=(f_{1}^{\downarrow k},\ldots,f_{l}^{\downarrow k})\]
let\[
\sum F=f_{1}+\ldots+f_{l}\]
and, for $g\in E_{\infty}^{\mathbb{C}}(x)$, let\[
gF=(gf_{1},\ldots,gf_{l})\]
\end{notation}
\begin{algorithm}
\label{alg:AsymptoticSolutions}(AsymptoticSolutions)\\
Input: $P(x,y)=a_{n}(x)y^{n}+\ldots+a_{0}(x)\in E_{\infty}^{\mathbb{C}}(x)[y]$
with $a_{n}\neq0$ and $a_{0}\neq0$, $m\in\mathbb{Z}_{>0},$ $\operatorname{sflag}\in\{\operatorname{true},\operatorname{false}\}$\\
Output: $((F_{1},\sigma_{1}),\ldots,(F_{t},\sigma_{t}))$ such
that
\begin{itemize}
\item for $1\leq i\leq t$\emph{,} $F_{i}=(f_{i,1},\ldots,f_{i,m_{i}})$
is an $m_{i}$-term asymptotic approximation of $\sigma_{i}$ roots
of $P(x,y)$ in $y$ (counted with multiplicities),
\item $f_{i,1},\ldots,f_{i,m_{i}}\in E_{\infty}^{\mathbb{C}}(x)$,
\item either $m_{i}=m$ or $m_{i}<m$ and $f_{i,1}+\ldots+f_{i,m_{i}}$
is an exact root of $P$ of multiplicity $\sigma_{i}$,\emph{ }
\item if $\operatorname{sflag}=\operatorname{false}$ then $\sigma_{1}+\ldots+\sigma_{t}=n$
and $F_{1},\ldots,F_{t}$ are asymptotic approximations of all complex
roots of $P(x,y)$,
\item if $\operatorname{sflag}=\operatorname{true}$ then $F_{1},\ldots,F_{t}$
are asymptotic approximations of all asymptotically small complex
roots of $P(x,y)$.\end{itemize}
\begin{enumerate}
\item If $P$ does not depend on $x$ then 

\begin{enumerate}
\item if $\operatorname{sflag}=\operatorname{true}$ return $()$,
\item let $r_{1},\ldots,r_{t}\in\mathbb{C}$ be the distinct roots of $P$,
\item for $1\leq i\leq t$\emph{,} let $\sigma_{i}$ be the multiplicity
of $r_{i}$,
\item return $(((r_{1}),\sigma_{1}),\ldots,((r_{t}),\sigma_{t}))$.
\end{enumerate}
\item Apply Algorithm \ref{alg:MrvApproxC} to $a_{0},\ldots,a_{n}$, obtaining
$\omega\in E_{\infty}(x)$,\[
b_{0},\ldots,b_{n}\in E_{\infty}^{\mathbb{C}}(x)\]
$e_{0},\ldots,e_{n}\in\mathbb{R}\cup\{\infty\}$, $d>0$, and $k\in\mathbb{Z}_{\geq0}$.
\item Let\textup{ $I=\{i\,:\,0\leq i\leq n\wedge b_{i}\neq0\}$} and \textup{$A=\{(i,e_{i})\::\: i\in I\}$}.
Compute $\gamma_{1},\ldots,\gamma_{l}$ such that the lower part of
the boundary of the convex hull of $A$ consists of segments with
slopes $-\gamma_{1},\ldots,-\gamma_{l}$. 
\item Set $\mathcal{R}=()$.
\item For $1\leq j\leq l$ do:

\begin{enumerate}
\item if $\operatorname{sflag}=\operatorname{true}$ and $\gamma_{j}\leq0$,
continue the loop with the next $j$,
\item compute \textup{$\mu=\min_{i\in I}e_{i}+\gamma_{j}i$ , $J=\{i\in I\,:\, e_{i}+\gamma_{j}i=\mu\}$,}
\item let \textup{$\lambda=\min J$ and let} $R(x,z)=\sum_{i\in J}b_{i}z^{i-\lambda}$,
\item compute\begin{align*}
 & ((F_{1},\sigma_{1}),\ldots,(F_{t},\sigma_{t}))=\\
 & AsymptoticSolutions(R(x,z),m,\operatorname{false})\end{align*}
where $F_{\iota}=(f_{\iota,1},\ldots,f_{\iota,m_{\iota}})$, for $1\leq\iota\leq t$,
\item For $1\leq\iota\leq t$ do:

\begin{enumerate}
\item put $G=(\omega^{\gamma_{j}}F)^{\downarrow k}$,
\item if $m_{\iota}=m$, set $\mathcal{R}=\mathcal{R}\sqcup((G,\sigma_{\iota}))$
and continue the loop with the next $\iota$,
\item if $m_{\iota}<m$, put $r=\sum G$, 
\item compute $P(x,r+(\omega^{\gamma_{j}})^{\downarrow k}y)=a_{r,n}(x)y^{n}+\ldots+a_{r,0}(x)$,
with $a_{r,i}\in E_{\infty}^{\mathbb{C}}(x)$, and let $\lambda=\min\{i\,:\, a_{r,i}\neq0\}$,
\item if $\lambda>0$, set $\mathcal{R}=\mathcal{R}\sqcup((G,\lambda))$,
and if $\lambda=\sigma_{\iota}$, continue the loop with the next
$\iota$,
\item put $P_{r}(x,y)=a_{r,n}(x)y^{n-\lambda}+\ldots+a_{r,\lambda}(x)$,
\item compute\begin{align*}
 & ((G_{1},\tau_{1}),\ldots,(G_{s},\tau_{s}))=\\
 & AsymptoticSolutions(P_{r}(x,y),m-m_{\iota},\operatorname{true})\end{align*}

\item for $1\leq\kappa\leq s$,
\item set\begin{align*}
 & \mathcal{R}=\mathcal{R}\sqcup\\
 & ((G\sqcup(\omega^{\gamma})^{\downarrow k}G_{1},\tau_{1}),\ldots,(G\sqcup(\omega^{\gamma})^{\downarrow k}G_{s},\tau_{s}))\end{align*}

\end{enumerate}
\end{enumerate}
\item Return $\mathcal{R}$.
\end{enumerate}
\end{algorithm}
\begin{proof}
For a proof of termination of the algorithm, put $\Sigma=\sum_{i=0}^{n}Size(a_{i})$
and let us use pairs $(m,\Sigma)\in\mathbb{Z}_{>0}\times\mathbb{Z}_{\geq0}$
as a metric. Note that $\mathbb{Z}_{>0}\times\mathbb{Z}_{\geq0}$
with the lexicographic order does not admit infinite strictly decreasing
sequences. The recursive calls to Algorithm \ref{alg:AsymptoticSolutions}
in step $5(d)$ have the same value of $m$ and a strictly lower value
of $\Sigma$, because according to the specification of Algorithm
\ref{alg:MrvApproxC}, $\sum_{i=0}^{n}Size(b_{i})<\sum_{i=0}^{n}Size(a_{i})$.
The recursive calls to Algorithm \ref{alg:AsymptoticSolutions} in
step $5(e)(vii)$ have a strictly lower value of $m$. This shows
that the algorithm terminates.

Correctness of the algorithm follows from the discussion earlier in
this section. \end{proof}
\begin{example}
Find one-term asymptotic approximations of roots of\[
P=y^{5}-\exp(x)y^{4}+x\exp(\pi x)y^{3}+\log(x)y-x^{2}\]
Applying Algorithm \ref{alg:MrvApproxC} to the coefficients of $P$
we get $\omega=\exp(-x)$,\begin{align*}
 & (b_{0},b_{1},b_{2},b_{3},b_{4},b_{5})=(-x^{2},\log(x),0,x,-1,1)\\
 & (e_{0},e_{1},e_{2},e_{3},e_{4},e_{5})=(0,0,\infty,-\pi,-1,0)\end{align*}
$d=\infty$, and $k=0$ (the coefficients and exponents can be read
off $P$ written in terms of $\omega$: $y^{5}-\omega^{-1}y^{4}+x\omega^{-\pi}y^{3}+\log(x)y-x^{2}$;
in this case $a_{i}^{\uparrow k}=b_{i}\omega^{e_{i}}$ hence $d=\infty$).
The set $A$ and the lower part of the boundary of the convex hull
of $A$ are shown in Figure \ref{fig:NewtonPolygon}. We obtain $\gamma_{1}=\frac{\pi}{3}$
and $\gamma_{2}=-\frac{\pi}{2}$.

For $\gamma_{1}$ we get $R_{1}=xy^{3}-x^{2}$. In the recursive call
to Algorithm \ref{alg:AsymptoticSolutions} applying Algorithm \ref{alg:MrvApproxC}
to the coefficients of $R_{1}$ yields $\omega_{1}=\exp(-x)$,\begin{align*}
 & (b_{1,0},b_{1,1},b_{1,2},b_{1,3})=(-1,0,0,1)\\
 & (e_{1,0},e_{1,1},e_{1,2},e_{1,3})=(-2,\infty,\infty,-1)\end{align*}
$d_{1}=\infty$, and $k_{1}=1$. The lower part of the boundary of
the convex hull of $A_{1}$ consists of one segment and we get $\gamma_{1,1}=-\frac{1}{3}$
and the corresponding polynomial $R_{1,1}=y^{3}-1$. The recursive
call to Algorithm \ref{alg:AsymptoticSolutions} returns simple roots
$1,\frac{-1-\imath\sqrt{3}}{2},\frac{-1+\imath\sqrt{3}}{2}$ of $R_{1,1}$.
We have $(\omega_{1}^{\gamma_{1,1}})^{\downarrow1}=\sqrt[3]{x}$ hence
Algorithm \ref{alg:AsymptoticSolutions} for $R_{1}$ returns\[
\sqrt[3]{x},\frac{-1-\imath\sqrt{3}}{2}\sqrt[3]{x},\frac{-1+\imath\sqrt{3}}{2}\sqrt[3]{x}\]
 all with multiplicity $1$. Since $(\omega^{\gamma_{1}})^{\downarrow0}=\exp(-\frac{\pi}{3}x)$,
we add\[
\sqrt[3]{x}\exp(-\frac{\pi}{3}x),\frac{-1-\imath\sqrt{3}}{2}\sqrt[3]{x}\exp(-\frac{\pi}{3}x),\frac{-1+\imath\sqrt{3}}{2}\sqrt[3]{x}\exp(-\frac{\pi}{3}x)\]
 to $\mathcal{R}$.

For $\gamma_{2}$ we get $R_{2}=y^{2}+x$. In the recursive call to
Algorithm \ref{alg:AsymptoticSolutions} applying Algorithm \ref{alg:MrvApproxC}
to the coefficients of $R_{2}$ yields $\omega_{2}=\exp(-x)$, $(b_{2,0},b_{2,1},b_{2,2})=(1,0,1)$,
$(e_{2,0},e_{2,1},e_{2,2})=(-1,\infty,0)$, $d_{2}=\infty$, and $k_{2}=1$.
The lower part of the boundary of the convex hull of $A_{2}$ consists
of one segment and we get $\gamma_{2,1}=-\frac{1}{2}$ and the corresponding
polynomial $R_{2,1}=y^{2}+1$. The recursive call to Algorithm \ref{alg:AsymptoticSolutions}
returns simple roots $-\imath,\imath$ of $R_{2,1}$. We have $(\omega_{2}^{\gamma_{2,1}})^{\downarrow1}=\sqrt{x}$
hence Algorithm \ref{alg:AsymptoticSolutions} for $R_{2}$ returns
$-\imath\sqrt{x},\imath\sqrt{x}$, both with multiplicity $1$. Since
$(\omega^{\gamma_{2}})^{\downarrow0}=\exp(\frac{\pi}{2}x)$, we add
$-\imath\sqrt{x}\exp(\frac{\pi}{2}x),\imath\sqrt{x}\exp(\frac{\pi}{2}x)$
to $\mathcal{R}$.

Finally, the algorithm returns one-term asymptotic approximations\begin{align*}
\mathcal{R}= & (\sqrt[3]{x}\exp(-\frac{\pi}{3}x),\frac{-1-\imath\sqrt{3}}{2}\sqrt[3]{x}\exp(-\frac{\pi}{3}x),\\
 & \frac{-1+\imath\sqrt{3}}{2}\sqrt[3]{x}\exp(-\frac{\pi}{3}x),-\imath\sqrt{x}\exp(\frac{\pi}{2}x),\imath\sqrt{x}\exp(\frac{\pi}{2}x))\end{align*}
all with multiplicity $1$.

\begin{figure}

\caption{\label{fig:NewtonPolygon}The Newton polygon of $A$}
\includegraphics[width=\columnwidth, trim = 0mm 0mm 0mm 0mm, clip]{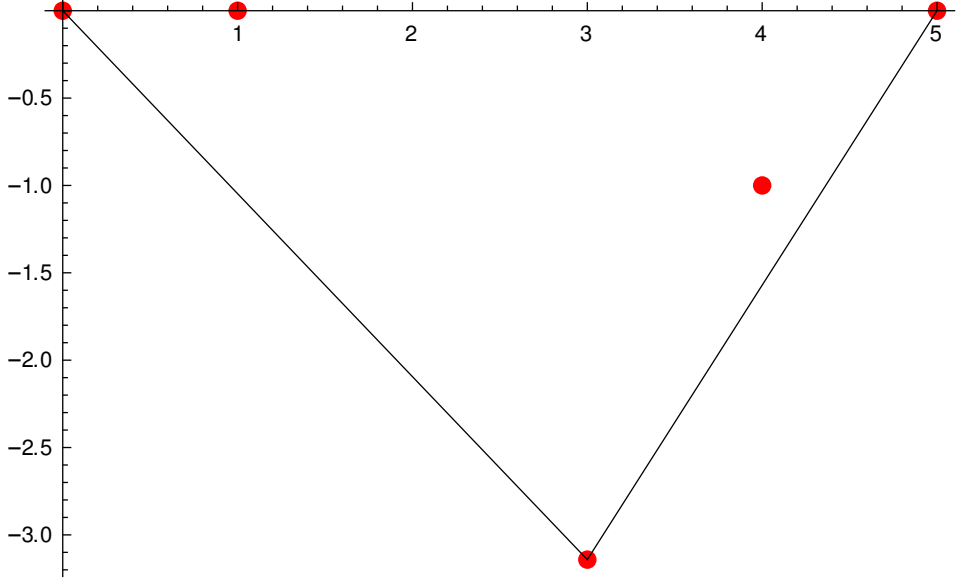}

\end{figure}

\end{example}

\section{Real roots}

Imaginary part of an non-real root may be asymptotically smaller than
the real part, hence real asymptotic approximations can correspond
to non-real roots.
\begin{example}
\label{exa:Real}Let $P(x,y)=(y^{2}-x\exp(x)y+\exp(2x))^{2}+1$. Three-term
asymptotic approximations of roots of $P$ in $y$ computed with Algorithm
\ref{alg:AsymptoticSolutions} are $(x^{-1}+x^{-3}+2x^{-5})\exp(x)$
and $(x-x^{-1}-x^{-3})\exp(x)$, both with multiplicity two. The approximations
are real-valued, yet $P$ clearly does not have real roots. Computing
more terms only adds real-valued terms of the form $ax^{-n}$ in the
coefficient of $\exp(x)$. Imaginary parts would show up only in a
transfinite series representation, since the imaginary parts are asymptotically
smaller than $x^{-n}\exp(x)$ for any $n$. For this low degree polynomial
$P$ we can compute asymptotic approximations of imaginary parts of
roots of $P$, by computing $R=res_{z}(P(x,y+z),P(x,z))$. The roots
of this polynomial of degree $16$ in $y$ with $32$ terms, are differences
of pairs of roots of $P$. In particular, four of the roots are equal
to the imaginary parts of roots of $P$, multiplied by two. One-term
asymptotic approximations of roots of $R$ computed with Algorithm
\ref{alg:AsymptoticSolutions} include two purely imaginary-valued
expressions, $-2\iota x^{-1}\exp(-x)$ and $2\iota x^{-1}\exp(-x)$,
both of multiplicity two. This shows that, indeed, the imaginary parts
of roots of $P$ are asymptotically smaller than $x^{-n}\exp(x)$
for any $n$. 
\end{example}
In this section we provide a method for deciding which asymptotic
approximations correspond to real roots. 

Let $P(x,y)=a_{n}(x)y^{n}+\ldots+a_{0}(x)\in E_{\infty}(x)[y]$ with
$a_{n}\neq0$ and $a_{0}\neq0$, and let $((F_{1},\sigma_{1}),\ldots,(F_{t},\sigma_{t}))$
be the output of Algorithm \ref{alg:AsymptoticSolutions} for $P$,
with some $m>0$ and $sflag=false$. Note that here we assume that
the coefficients of $P$ are real-valued. For $1\leq i\leq t$, $F_{i}=(f_{i,1},\ldots,f_{i,m_{i}})$.
First, let us note the easy cases. If any of $f_{i,j}$ are not real-valued,
then $F_{i}$ does not correspond to a real root. If all $f_{i,j}$
are real-valued and $\sigma_{i}=1$ then $F_{i}$ corresponds to a
real root. If $m_{i}<m$ then $f_{i,1}+\ldots+f_{i,m_{i}}$ is equal
to the exact root, hence it is evident whether the root is real valued.

The hard case is when there are real-valued asymptotic approximations
with $\sigma_{i}>1$ and $m_{i}=m$. We will find the number of distinct
real roots corresponding to each real-valued asymptotic approximation
with multiplicity higher than one. Assume that, possibly after reordering,
the real-valued asymptotic approximations are $F_{i}$, for $1\leq i\leq s$.
Let $r_{i}=f_{i,1}+\ldots+f_{i,m_{i}}$, for $1\leq i\leq s$. The
algorithm \emph{MrvLimit} of \cite{G} contains a subprocedure which
computes the sign of exp-log expressions near infinity. Hence, we
can reorder the approximations so that, for sufficiently large $x$,
$r_{1}<\ldots<r_{s}$. Put $h_{0}=-\infty$, $h_{i}=\frac{r_{i}+r_{i+1}}{2}$,
for $1\leq i<s$, and $h_{s}=\infty$.
\begin{lem}
\label{lem:inequality}If $F_{i}$ corresponds to a real root $\alpha$
of $P$ then, for sufficiently large $x$, $h_{i-1}<\alpha<h_{i}$. 
\end{lem}
To prove the lemma we will use the following claim.
\begin{claim}
\label{claim}If $F=(f_{1},\ldots,f_{m_{f}})$ and $G=(g_{1},\ldots,g_{m_{g}})$
are asymptotic approximations returned by Algorithm \ref{alg:AsymptoticSolutions}
for $P$ and there is $l\leq\min(m_{f},m_{g})$ such that, for all
$1\leq j<l$, $f_{j}=g_{j}$ and $f_{l}\neq g_{l}$, then $\lim_{x\rightarrow\infty}\frac{f_{l}}{g_{l}}\neq1$.\end{claim}
\begin{proof}
The claim is true when $P$ does not depend on $x$, hence, by induction,
we may assume that the claim is true for the recursive calls to Algorithm
\ref{alg:AsymptoticSolutions}. If $F$ and $G$ were computed in
different iterations of the loop in step $5$ then $l=1$, $f_{1}=(\omega^{\gamma_{f}}c_{f})^{\downarrow k}$,
$g_{1}=(\omega^{\gamma_{g}}c_{g})^{\downarrow k}$, $mrv(c_{f},c_{g})\prec\omega$,
and $\gamma_{f}\neq\gamma_{g}$, hence $\lim_{x\rightarrow\infty}\frac{\lvert f_{l}\rvert}{\lvert g_{l}\rvert}$
is either $0$ or $\infty$. If $F$ and $G$ were computed in the
same iteration of the loop in step $5$, but in different iterations
of the loop in step $5(e)$ then the claim is true by the inductive
hypothesis applied to $R(x,z)$. Finally, if $F$ and $G$ were computed
in the same iteration of the loop in step $5(e)$ then $l>m_{\iota}$,
and hence neither $F$ nor $G$ was added in step $5(e)(v)$. Therefore,
the claim is true by the inductive hypothesis applied to $P_{r}(x,y)$.
\end{proof}
Let us now prove Lemma \ref{lem:inequality}. 
\begin{proof}
If $m_{i}<m$ then $\alpha=r_{i}$ and $h_{i-1}<r_{i}<h_{i}$, because,
for sufficiently large $x$, $r_{1}<\ldots<r_{s}$. Hence we can assume
that $m_{i}=m$. Let us prove that $h_{i-1}<\alpha$. If $i=1$, then
$h_{i-1}=-\infty$ and the inequality is true. Let $l$ be such that
for all $1\leq j<l$ $f_{i-1,j}=f_{i,j}$ and $f_{i-1,l}\neq f_{i,l}$.
We define $f_{i-1,m_{i-1}+1}=0$, so that such $l$ always exist.
Note that, for sufficiently large $x$, $f_{i-1,l}<f_{i,l}$, because
$r_{i-1}<r_{i}$. If $\lim_{x\rightarrow\infty}\frac{\lvert f_{i-1,l}\rvert}{\lvert f_{i,l}\rvert}\leq1$
put $g=f_{i,l}$ else put $g=f_{i-1,l}$. We have\[
\frac{\alpha-h_{i-1}}{g}=\frac{\alpha-\sum_{j=1}^{m}f_{i,j}}{f_{i,m}}\frac{f_{i,m}}{g}+\frac{f_{i,l}-f_{i-1,l}}{2g}+\frac{\sum_{j=l+1}^{m}f_{i,j}}{2g}-\frac{\sum_{j=l+1}^{m_{i-1}}f_{i-1,j}}{2g}\]
Since $F_{i}$ is an asymptotic approximation of $\alpha$, we have
\[
\lim_{x\rightarrow\infty}\frac{\alpha-\sum_{j=1}^{m}f_{i,j}}{f_{i,m}}=0\]
and, for $j>l$, we have $\lim_{x\rightarrow\infty}\frac{f_{i,j}}{g}=0$
and $\lim_{x\rightarrow\infty}\frac{f_{i-1,j}}{g}=0$, therefore\[
\lim_{x\rightarrow\infty}\frac{\alpha-h_{i-1}}{g}=\lim_{x\rightarrow\infty}\frac{f_{i,l}-f_{i-1,l}}{2g}\]

If $f_{i-1,l}=0$, then $g=f_{i,l}$, for sufficiently large $x$,
$f_{i,l}>0$, and $\lim_{x\rightarrow\infty}\frac{f_{i,l}-f_{i-1,l}}{2g}=\frac{1}{2}$,
hence for sufficiently large $x$, $\alpha-h_{i-1}>0$. 

If $f_{i-1,l}\neq0$, then the assumptions of Claim \ref{claim} are
satisfied, and hence $\lim_{x\rightarrow\infty}\frac{f_{i-1,l}}{f_{i,l}}\neq1$. 

Suppose that $g=f_{i,l}$. Then\[
\lim_{x\rightarrow\infty}\frac{f_{i,l}-f_{i-1,l}}{2g}=\frac{1}{2}-\frac{1}{2}\lim_{x\rightarrow\infty}\frac{f_{i-1,l}}{f_{i,l}}\neq0\]
Since for sufficiently large $x$, $f_{i,l}-f_{i-1,l}>0$, $\lim_{x\rightarrow\infty}\frac{\lvert f_{i-1,l}\rvert}{\lvert f_{i,l}\rvert}\leq1$,
and $\lim_{x\rightarrow\infty}\frac{f_{i-1,l}}{f_{i,l}}\neq1$, hence,
for sufficiently large $x$, $f_{i,l}>0$. Therefore,\[
\lim_{x\rightarrow\infty}\frac{\alpha-h_{i-1}}{f_{i,l}}=\lim_{x\rightarrow\infty}\frac{f_{i,l}-f_{i-1,l}}{2f_{i,l}}>0\]
 which shows that, for sufficiently large $x$, $h_{i-1}<\alpha$. 

Now suppose that $g=f_{i-1,l}$. Then \[
\lim_{x\rightarrow\infty}\frac{f_{i,l}-f_{i-1,l}}{2g}=\frac{1}{2}\lim_{x\rightarrow\infty}\frac{f_{i,l}}{f_{i-1,l}}-\frac{1}{2}\neq0\]
Since for sufficiently large $x$, $f_{i,l}-f_{i-1,l}>0$, and $\lim_{x\rightarrow\infty}\frac{\lvert f_{i-1,l}\rvert}{\lvert f_{i,l}\rvert}>1$,
hence, for sufficiently large $x$, $f_{i-1,l}<0$. Therefore, \[
\lim_{x\rightarrow\infty}\frac{\alpha-h_{i-1}}{f_{i-1,l}}=\lim_{x\rightarrow\infty}\frac{f_{i,l}-f_{i-1,l}}{2f_{i-1,l}}<0\]
which shows that, for sufficiently large $x$, $h_{i-1}<\alpha$.
The proof that, for sufficiently large $x$, $\alpha<h_{i}$ is similar.
\end{proof}
Let $P_{1},\ldots,P_{k}$ be the Sturm sequence of $P$ in $y$ over
$E_{\infty}(x)$ (that is coefficients that are identically zero near
infinity are set to zero). For $1\leq i\leq s-1$ and $1\leq j\leq k$,
$P_{j}(x,h_{i})\in E_{\infty}(x)$, hence it has a constant sign $\theta_{i,j}$
near infinity, and we can compute $\theta_{i,j}$ using a subprocedure
of \emph{MrvLimit.} Let $c_{j}y^{n_{j}}$ be the leading term of $P_{j}$,
for $1\leq j\leq k$, let $\theta_{0,j}$ be the sign near infinity
of $(-1)^{n_{j}}c_{j}$, and let $\theta_{s,j}$ be the sign near
infinity of $c_{j}$. For $1\leq i\leq s-1$, let $\nu_{i}$ be the
number of sign changes in the sequence $\Theta_{i}=(\theta_{i,1},\ldots,\theta_{i,k})$.
\begin{criterion}
The number of distinct real roots of $P$ corresponding to the asymptotic
approximation $F_{i}$ is equal to $\nu_{i-1}-\nu_{i}$.
\end{criterion}
Correctness of the criterion follows from Lemma \ref{lem:inequality}
and Sturm's theorem.
\begin{example}
As in Example \ref{exa:Real}, let $P(x,y)=(y^{2}-x\exp(x)y+\exp(2x))^{2}+1$.
One-term asymptotic approximations of roots of $P$ in $y$ computed
with Algorithm \ref{alg:AsymptoticSolutions} are $r_{1}=x^{-1}\exp(x)$
and $r_{2}=x\exp(x)$, both with multiplicity two. The Sturm sequence
of $P$ in $y$ is\begin{eqnarray*}
P_{1} & = & y^{4}-2x\exp(x)y^{3}+(x^{2}+2)\exp(2x)y^{2}-2x\exp(3x)y+\\
 &  & \exp(4x)+1\\
P_{2} & = & 4y^{3}-6x\exp(x)y^{2}+2(x^{2}+2)\exp(2x)y-2x\exp(3x)\\
P_{3} & = & \frac{1}{4}(x^{2}-4)\exp(2x)y^{2}-\frac{1}{4}(x^{3}-4x)\exp(3x)y+\\
 &  & \frac{1}{4}(x^{2}-4)\exp(4x)-1\\
P_{4} & = & -\frac{16}{(x^{2}-4)\exp(2x)}y+\frac{8x}{(x^{2}-4)\exp(x)}\\
P_{5} & = & \frac{1}{16}(x^{4}-8x^{2}+16)\exp(4x)+1\end{eqnarray*}
We have $\Theta_{0}=(1,-1,1,1,1)$ and $\Theta_{2}=(1,1,1,-1,1)$.
Since $\nu_{0}=\nu_{2}=2$, $P$ has no real roots near infinity (and
we do not need to compute $\nu_{1}$, since it must equal $2$ as
well). 

Let $Q(x,y)=(y^{2}-x\exp(x)y+\exp(2x))^{2}-1$. One-term asymptotic
approximations of roots of $Q$ in $y$ computed with Algorithm \ref{alg:AsymptoticSolutions}
are the same as for $P$. The Sturm sequence of $Q$ in $y$ is\begin{eqnarray*}
Q_{1} & = & y^{4}-2x\exp(x)y^{3}+(x^{2}+2)\exp(2x)y^{2}-2x\exp(3x)y+\\
 &  & \exp(4x)-1\\
Q_{2} & = & 4y^{3}-6x\exp(x)y^{2}+2(x^{2}+2)\exp(2x)y-2x\exp(3x)\\
Q_{3} & = & \frac{1}{4}(x^{2}-4)\exp(2x)y^{2}-\frac{1}{4}(x^{3}-4x)\exp(3x)y+\\
 &  & \frac{1}{4}(x^{2}-4)\exp(4x)+1\\
Q_{4} & = & \frac{16}{(x^{2}-4)\exp(2x)}y-\frac{8x}{(x^{2}-4)\exp(x)}\\
Q_{5} & = & \frac{1}{16}(x^{4}-8x^{2}+16)\exp(4x)-1\end{eqnarray*}
We have $\Theta_{0}=(1,-1,1,-1,1)$ and $\Theta_{2}=(1,1,1,1,1)$.
Since $\nu_{0}=4$ and $\nu_{2}=0$, $P$ has four distinct real roots
near infinity. This again is sufficient to tell that $r_{1}$ or $r_{2}$
correspond to two real root each. And indeed, if we substitute $h_{1}=(r_{1}+r_{2})/2=(x^{-1}+x)\exp(x)$
into the Sturm sequence and compute the signs near infinity we get
$\Theta_{1}=(1,-1,-1,1,1)$ and $\nu_{1}=2$.
\end{example}

\section{Implementation and experimental results}

We have implemented \emph{AsymptoticSolutions} as a part of the \emph{Mathematica}
system. The implementation has been done in Wolfram Language, using
elements of the \emph{MrvLimit} algorithm, which is implemented partly
in the C source code of \emph{Mathematica} and partly in Wolfram Language.
The experiments have been run on a laptop computer with a $2.7$ GHz
Intel Core i7-4800MQ processor and $10$ GB of RAM assigned to the
Linux virtual machine.
\begin{example}
\label{exa:Experiment}We use eight exp-log expressions from examples
in \cite{G} as polynomial coefficients.\begin{eqnarray*}
a_{0} & = & e^{x}(e^{1/x-e^{-x}}-e^{1/x})\\
a_{1} & = & e^{\frac{e^{x-e^{-x}}}{1-1/x}}-e^{e^{x}}\\
a_{2} & = & \frac{e^{e^{e^{x+e^{-x}}}}}{e^{e^{e^{x}}}}\\
a_{3} & = & \frac{e^{e^{e^{x}}}}{e^{e^{e^{x-e^{-e^{x}}}}}}\\
a_{4} & = & (3^{x}+5^{x})^{1/x}\\
a_{5} & = & \frac{x}{\log(x^{\log(x)^{\log(2)/\log(x)}})}\\
a_{6} & = & \frac{\exp(4xe^{-x}/(e^{-x}+e^{\frac{-2x^{2}}{x+1}}))-e^{x}}{e^{4x}}\\
a_{7} & = & \frac{\exp(\frac{xe^{-x}}{e^{-x}+e^{-2x^{2}/(x+1)}})}{e^{x}}\end{eqnarray*}
Let $P_{n}(x,y)=\sum_{i=0}^{n}a_{i}y^{i}$. We have run the examples
with $n$ ranging from $2$ to $7$ and with varying number $m$ of
requested terms. The results are given in Table \ref{tab:Experiment}.
For each value of $m$ the row \emph{Time} gives the computation time
in seconds, the row \emph{Iter} gives the number of calls to Algorithm
\ref{alg:AsymptoticSolutions}, and the row \emph{LC} gives the total
leaf count of the returned expressions.

We can observe that for a fixed polynomial the number of recursive
calls is close to linear in the number of additional terms requested.
Increasing the degree did not necessarily lead to higher complexity,
e.g. adding the degree $6$ term made the computation easier. A likely
cause for this is that the dominating terms in the degree $6$ polynomial
were simpler than those in the degree $5$ polynomial.
\end{example}

\vfill
\break

\begin{table}
\caption{\label{tab:Experiment}Example \ref{exa:Experiment}.}

\begin{tabular}{|c|c||c|c|c|c|c|c|}
\hline 
$m$ & $n$ & $2$ & $3$ & $4$ & $5$ & $6$ & $7$\tabularnewline
\hline
\hline 
 & Time & $0.171$ & $0.235$ & $0.292$ & $0.276$ & $0.389$ & $0.448$\tabularnewline
\cline{2-8} 
$1$ & Iter & $4$ & $6$ & $6$ & $7$ & $6$ & $6$\tabularnewline
\cline{2-8} 
 & LC & $56$ & $84$ & $124$ & $195$ & $192$ & $323$\tabularnewline
\hline 
 & Time & $0.231$ & $0.508$ & $1.18$ & $1.20$ & $1.02$ & $1.87$\tabularnewline
\cline{2-8} 
$5$ & Iter & $20$ & $32$ & $44$ & $59$ & $58$ & $67$\tabularnewline
\cline{2-8} 
 & LC & $136$ & $229$ & $498$ & $668$ & $588$ & $1933$\tabularnewline
\hline 
 & Time & $0.376$ & $0.796$ & $3.53$ & $5.25$ & $1.64$ & $5.18$\tabularnewline
\cline{2-8} 
$10$ & Iter & $40$ & $63$ & $86$ & $124$ & $118$ & $137$\tabularnewline
\cline{2-8} 
 & LC & $256$ & $433$ & $896$ & $1303$ & $1188$ & $4008$\tabularnewline
\hline 
 & Time & $0.968$ & $3.21$ & $26.5$ & $25.4$ & $5.47$ & $20.1$\tabularnewline
\cline{2-8} 
$20$ & Iter & $80$ & $124$ & $170$ & $254$ & $238$ & $277$\tabularnewline
\cline{2-8} 
 & LC & $496$ & $837$ & $1716$ & $2573$ & $2388$ & $8158$\tabularnewline
\hline
\end{tabular}%
\end{table}

\bibliographystyle{abbrv}

\end{document}